\documentclass{article}

\usepackage{amsmath, amsthm, amssymb}
\usepackage{hyperref} 
\usepackage{cleveref}
\usepackage{apxproof}
\usepackage{xcolor} 

\newif\ifLNCS

\newif\ifLipics

\newif\ifShortVersion

\newcommand{\FFD}{$\mathsf{FFD}$ }

\newcommand{\kbp}{$k$BP }
\newcommand{\qbp}{$q$BP }

\newcommand{\ffd}{$\mathsf{FFD}$ }
\newcommand{\ffdq}{$\mathsf{FFDq}$ }
\newcommand{\ffda}{$\mathsf{FFD}({D})$ }
\newcommand{\ffdqa}{${\mathsf{FFDq}({D'}_q)}$ }
\newcommand{\opt}{$\mathsf{OPT}$ }
\newcommand{\opta}{$\mathsf{OPT}(D)$ }
\newcommand{\optqa}{$\mathsf{OPT}({D}_q)$ }

\newcommand{\dbpi}{$d_{B_P}^i$ }
\newcommand{\wdbpi}{$W(d_{B_P}^i)$ }
\newcommand{\calb}{$\mathcal{B}$ }

\newcommand{\wcalbp}{$W(\mathcal{B}_P)$ }
\newcommand{\wcalbpi}[1]{$W(#1,\mathcal{B}_P)$ }

\newtheoremrep{theorem}{Theorem}[section]
\newtheoremrep{definition}{Definition}[section]
\newtheoremrep{lemma}{Lemma}[theorem]
\newtheoremrep{proposition}{Proposition}[theorem]
\newtheoremrep{corollary}{Corollary}[lemma]
\newtheoremrep{lemmar}[theorem]{Lemma}
\newtheoremrep{example}{Example}[section]

\newcommand{\dk}[1]{\textcolor{blue}{#1}}
\newcommand{\dkb}[1]{\dk{Dinesh says: #1}}

\newcommand{\shortabstract}{
	we study the \FFD algorithm for \kbp.
}

\newcommand{\longabstract}{
	We consider the bin packing problem with constraint on the number of copies of each item and their placements. 
}

\title{On the upper bound of the generalization of \FFD  to solve \qbp  for some special cases \\[0.5em] \large \textbf{Working Paper} }

\author{Dinesh Kumar Baghel \\
\large School of Computer Science, UPES, Dehradun - 248007 \\
\texttt{dinkubag21@gmail.com}}

\begin{document}	
	\maketitle

%
%
%

	\begin{abstract}	
		We consider a variant of the bin packing problem with constraints on the number of copies of each item and their placement in the packing.
		The input 
		$D_q := DD\ldots$ is defined as $q$ consecutive copies of the multiset $D$, 
		with a fixed bin capacity $S$.
		Note that, for each item in $D$, there are $q$ copies in $D_q$.
		The goal is to pack all the items in $D_q$ into the minimum number of bins, such that each bin contains at most one copy of each item and the total size of all items in a bin does not exceed the bin capacity $S$. We call this problem $q$BP.
		
		First Fit Decreasing (\FFD) is a classical bin packing algorithm: it first orders the items in nonincreasing order, then packs the next item into the first bin where it fits. 
		In the literature, \FFD proofs rely on the assumption that the last bin in the \FFD packing contains only a single item. This assumption does not naturally extend to the $q$BP problem.
		In this paper, we circumvent this difficulty by analyzing \ffdq on a carefully chosen subinstance 
		${D'}_q \subseteq D_q$ ($q$ consecutive copies of $D$, each copy sorted in non-increasing order) while preserving the same upper bound for the original input $D_q$. 
		We show that
		the approximation ratio of 
		\ffdq for some special cases is
		\begin{align*}
			\mathsf{FFDq(D_q)} \leq \frac{11}{9}\mathsf{OPT(D_q)} + 3q
		\end{align*}
		where $\mathsf{FFDq}$ and $\mathsf{OPT}$ denote the number of bins used by the \FFD generalization and by an optimal algorithm, respectively.
		
	\end{abstract}

	\tableofcontents
	

\section{Introduction}
\label{sec: introduction}
We define \emph{$q$-times bin packing} (or \qbp in short) as: given a multiset of items of different sizes, a fixed bin capacity, and an integer $q \geq 1$. 
The \qbp problem is to pack these items into a minimum number of bins such that the sum of the sizes of the items in a bin does not exceed the bin capacity.
The packing must satify the following constraints:
\begin{enumerate}
	\item[C1]: each item appears at most once in a bin, 
	and
	
	\item[C2]: each item appears in $q$ different bins. We call these constraints as \qbp constraints.
\end{enumerate}

Formally, we consider a (multi)set $D$ of positive real numbers, a positive number $S$ that represents the bin capacity $S$, and an integer $q \geq 1$. Typically, $\sum d_i > S$ (where $d_i$ is the size of item $i$ in $D$). 
Given $q$, the goal is to pack the items in $D$ in a minimum number of bins such that the packing satisifies the \qbp constraints.

The problem \qbp has application in electricity distirbution \cite{baghel2026time, baghel2024k}.
Other application is creating a back up of files on different file servers (obviously there should not be more than one copy of the same file on the same server) \cite{Jansen_1998}.

Bin packing problem is $\mathsf{NP}$-hard \cite{GareyJohnson1979}. 
\FFD is a classical bin packing problem that works as follows:
first, the items are ordered by descending size, then in this order it attempts to pack the new item into the first bin, where it fits. In \FFD, all the bins are kept open in the order in which they are opened.

In their work, \cite{baghel2026time} modeled the electricity distribution problem as \qbp problem. They applied the generalization of \FFD  to solve \qbp as follows: items are packed using the classical \FFD algorithm while satisfying the constraint C1. 
The final packing of the input must satisy the constraints C1 and C2.
In their work, they explored the two ways in which an input can be given to the packing algorithms: \\
-- In the first way, an input contains $q$ copies of the first item, then $q$ copies of the next item, and so on. \\
-- In the second way, an input to the packing algorithm is $D_q$ which is defined as the $q$ consecutive copies of $D$ where each copy is sorted in non-increasing order.\\
They have argued that in the first way algorithm will generate the packing that contains the number of bins that is $q$ times the number of bins in the packing of $D$ when $q=1$, whereas the packing can improve if the input is considered in second way. 
Therefore, in this paper we consider the input as $D_q$ as described above. 
Also, they have given a lower bound of 
$\frac{7}{6}\mathsf{OPT}(D_q) + 1$ that is based on the example given in \cite{dosa2007tight}. They conjectured the upper bound to be 
$\frac{11}{9} \mathsf{OPT}(D_q) + \frac{6}{9}$, and conducted experiments that supports this upper bound.

Lets calls the generalization of \FFD to sovle \qbp as \ffdq. Our paper contributes in the following way:
\begin{itemize}
	\item As reported in \cite{baghel2026time}, the main hindrance in extending the upper bound proofs for \FFD lies in the asusmption that the last bin in the \FFD packing of $D$ contains only a single item. This assumption does not naturally hold for \qbp. 
	In this paper, we overcome this by analyzing \ffdq for input ${D'}_q \subseteq D_q$ without compromising the upper bound for the input $D_q$.
	
	
	\item We have extended the proof of \cite{Baker_1985}, and prove that for some simple cases (these cases depends on the size of the first item in the last bin in the \FFD packing of $D_q$) the asymptotic approximation ratio is:
	\begin{align*}
		\mathsf{FFDq(D_q)} \leq \frac{11}{9}\mathsf{OPT(D_q)} + 3q
	\end{align*}
\end{itemize}
To the best of our knowledge, we are the first to give upper bound on \ffdq.

The outline of the remaining article is as follows:
in \Cref{sec: literature} we discuss the literature relevant to \FFD.
In \Cref{sec: definitions and notations} we describe the notations and definitions that are required in the remaining article.
In \Cref{sec: ffdq upper bound special case}, we give the proof for the upper bound of the generalization for the \FFD algorithm to solve \qbp for some special cases.
In \Cref{sec: conclusion}, we conclude this article.
	\section{Literature Review}
\label{sec: literature}
In this section, we will cover the literature relevant to First Fit Decreasing Algorithm.
$
\mathsf{FFD}(D) \leq \frac{11}{9} \mathsf{OPT}(D) +4
$
, showing the asymptotic approxmation ratio of $11/9 \approx 1.222$. His proff required more than 100 pages.
Later work by Baker \cite{Baker_1985}, and Yue \cite{Minyi} improved the additive constant from $4$ to $3$, then to $1$.
Yue together with Li \cite{Li_Yue_1997} further reduced the additive constant from $1$ to $7/9$.

Dosa \cite{dosa2007tight, Dósa_Li_Han_Tuza_2013} showed that \FFD satisfies 
$
	\mathsf{FFD}(D) \leq \frac{11}{9} \mathsf{OPT}(D) + \frac{6}{9}
$
and that this inequality is tight, and settled a long standing question about the best possible additive constant.
	\section{Definitions and Notations}
\label{sec: definitions and notations}
Let $D$ be the set of items to pack. Each item in $D$ is in the range $(0,S]$ where $S$ is some natural number that denotes the bin capacity. 
We denote by $d_i$ the $i$th item in $D$.  Let \ffda and \opta denotes the \ffd and \opt packing of the input instance $D$. 
The input instance $D$ is divided into sets of pieces called as \emph{regions}. We denote by $R_i$ the $i$the region.

This paper is a generalization of the Baker paper \cite{Baker_1985}.
Let $P$ denote some packing of the items in $D$. We denote by $B_P$ a bin in packing $P$ ($B$, when the packing $P$ is clear in the context). Then, for some bin $B$ in $P$ $i$th largest item in $B$ is denoted as \dbpi. 

In this paper we are using a weighting function $W$ with domain $[x_s,S]$ where $x_s$ is the smallest item in instance $D$,and range $[0,1]$
 (it is the same function used in Baker's paper \cite{Baker_1985}). It is a nondecreasing step function. 

\begin{table}[ht]
	\centering
	\begin{tabular}{c|c|c}
		\hline
		$i$ & $R_i$ & $W_i$ \\ \hline
		$0$ & $[x_s,\frac{S}{3}]$ & $\frac{1}{3}$ \\ \hline
		$1$ & $(\frac{S}{3},\frac{S}{2}]$ & $\frac{1}{2}$ \\ \hline
		$2$ & $(\frac{S}{2},S]$ & $1$ \\ \hline
	\end{tabular}
	\caption{A weighting function}
	\label{tab:example weighting function}
\end{table}

Weighting functions are specified using a table for example \Cref{tab:example weighting function}.
The weight of an item $d_i$ is $W(d_i)$ and its \emph{density} is defined as $\frac{S \cdot W(d_i)}{d_i}$. Let $W(B_P)$ be the total weight of all items packed in a bin $B_P$. Similarly, we define \wdbpi as the weight of the $i$th item in $B_P$ (this weight is $0$ if $B_P$ do not have any $i$th item). 
Let \calb be the set of bins in $P$, we define \wcalbp as the sum of the weights of all the bins in \calb in $P$.
Define \wcalbpi{i} be the sum of the weight of all the $i$th item in the set of bins \calb in $P$.

Bins in the packing are ordered as $B_1, B_2, \ldots$ where the subscript denote the sequence in which the bins are open while packing the item in $D$ using \ffdq.

The input to the \ffdq algorithm is $D_q$ defined as $D_q = DD\ldots \text{ q}th \text{ copy of }D$, each copy of $D$ sorted in non-increasing order.
We denote by $x$ the first item in the last bin in the \ffdq packing of instance $D_q$.
	\section{\ffdq Upper Bound For Special Cases}
\label{sec: ffdq upper bound special case}

In this section, we will give a proof sketch of the main theorem of this paper. Before that we will discuss some necessary definitions that are required in this section.
\begin{definition} \label{def: large item}
	If the size of an item is $>S/2$ we call it as large item.
\end{definition}

\begin{definition} \label{def: regular item and fallback item}[Du \cite{Baker_1985}]
	An item $d_i$ is called a \emph{regular} item such that all bins numbered higher than the bin in which the item $d_i$ is packed are empty.
	Otherwise, the item is called as \emph{fallback} item.
	Clearly, $x$ is the first \emph{regular} item in the last bin in the \ffdq packing of $D_q$.
\end{definition}

\begin{definition} \label{def: regular bin and transition bin}[Due to \cite{Baker_1985}]
	A \emph{regular} bin (except the last bin or the bin with highest index) is defined as one in which all the items that are packed in it are from a single region $R_i$ (with respect to a weighting function).
	Otherwise, it is called as \emph{transition} bin.
\end{definition}

Existing upper bound proofs in the literature assumes that in an instance $D$, $x$ is the smallest item, and is also the first item in the last bin the \FFD packing of $D$.
This assumption is not directly applicable in \qbp when $q>1$. Consider the following example:
\begin{example}
	Let $D = \{30, 20, 15\}$ and $S=50$. Then, \ffdq packing (which is equivalent to \FFD packing) when $q=1$ is $\{30,20\}, \{15\}$ whereas for $q=2$ the packing is $\{30, 20\}, \{15, 30\}, \{20, 15\}$.
\end{example}
For the purpose of analyzing the upper bound,
it is clear that we need a way to get around this problem. To resolve this, for the purpose of analysis, we assume the input to \ffdq as:

\begin{definition}\label{def: D prime q}
	We define the instance ${D'}_q$ as follows:
	Let $D^\downarrow$ be the input instance sorted in non-increasing order. Then ${D'}_q$ is $DDD\ldots (q-1) \text{ copies of } D$ (each copy of $D$ sorted in non-increasing order) followed by a subset of $D^\downarrow$ that contains items from largest till item $x$.
\end{definition}


\begin{proposition} \label{prop: same ffdq packing}
	\ffdqa $= \mathsf{FFDq}(D_q)$ .
\end{proposition}
\begin{proof}
	The proof is easy. Since the remaining items, that is items in $D_q \setminus {D'}_q$ are either packed in  the previous bins or they are packed in the last bin. In both of the cases, they do not have any impact on the number bins required to pack the items in the \ffdq packing of 
	$D_q$ (or ${D'}_q$).
	Note that when $q=1$, this is the smallest piece in the original piece whereas for $q>1$ it need not to be so.
\end{proof} 

Now, we state the main theorem:
\begin{theorem} \label{thm: main theorem 11/9 bound}
	For every input instance $D_q$ and  for some natural $q \geq 1$, $\mathsf{FFDq}(D_q) = \mathsf{FFDq({D'}_q)}
	\leq \frac{11}{9}\mathsf{OPT}({D'}_q) + 3q
	\leq \frac{11}{9}\mathsf{OPT}(D_q) + 3q$.
\end{theorem}

Other assumptions that we can made is this:
All the \emph{large} items are packed in the same bin in both \ffdqa and $\mathsf{OPT}({D'}_q)$. 
Let ${L}$ be the set of bins that contains a large item. A bin in ${L}$ is called as a $L-$bin.

Now we present a generalization of Lemma 1 (given in \cite{Baker_1985}). This lemma holds true for every weighting function.

\begin{lemma}
	\label{lem: 2 item in optqa at most 2 and 3 item in ffdqa}
	Let $x > \frac{2}{11} \cdot S$ and let $W$ be a weighting function. This weighting function assigns a weight of $0$ to all items with size less than $x$. Then, the total weight of the second items (having non-zero weight) in $L$-bins of $\mathsf{OPT}({D'}_q)$ is at most the total weight of the second and third items (of non-zero weight) in $L$-bins of \ffdqa except at most $q-1$ items with total weight at most $\frac{q-1}{2}$. 
\end{lemma}
\begin{proof}
	The proof consists of two parts. 
	In the first part of the proof, we will prove that \emph{the number of second items of weight at least $W_i$ in $L$-bins of $\mathsf{OPT}({D'}_q)$ is at most the number of second and third items of weight at least $W_i$ in $L$-bins of \ffdqa, for $0 \leq i \leq m$.}
	In the second part, we use the results from the first part and prove that the total weight of the second items in $L$-bins in $\mathsf{OPT}({D'}_q)$ is at most the total weight of the second and third items in $L$-bins of \ffdqa.
	
	\paragraph{Part 1:}
	We use mathematical induction to prove the first part. Assume that $W$ has weights $W_i$ for $0 \leq i \leq m$.
	\paragraph{Base case: } Let $q=1$, the lemma holds by the proof of lemma 1 in \cite{Baker_1985}. 
	\paragraph{Induction hypothesis:} Let us assume that the lemma holds for $q-1$.
	\paragraph{Induction case: } We will prove that the lemma holds for $q$. 
	First, we will prove that the number of second items (of non-zero weight) with weight at least $W_i$ in $L$-bins in $\mathsf{OPT}({D'}_q)$ is at most the number of second and third items (of non-zero weight) with weight at least $W_i$ in $L$-bins of \ffdqa, for $0 \leq i \leq m$. 
	The argument of our proof generalizes that given in Baker's paper \cite{Baker_1985}.
	Recall that the last instance in ${D'}_q$ is $D' \subseteq D$. \ffdq first packs the items in ${D'}_{q-1}$, then scans the bins in a first-fit way, and packs the items in $D'$ in this existing packing. 
	Now, we will discuss that at most two items with size at most $\frac{1}{2}$ and weight at least $W_i$ can be excluded from our analysis. Note that some items from $D'$ can be packed in the last bin in the \ffdq packing of ${D'}_{q-1}$. The possible cases are:
	\begin{itemize}
		\item [K1] No large item 
		from instance $D'$ has been packed in the last bin in the \ffdq packing of ${D'}_{q-1}$.
		
		\item [K2] One large item has been packed in the last bin in the \ffdq packing of ${D'}_{q-1}$, and there are no items with size at most $\frac{S}{2}$ and weight at least $W_i$ in this bin. In this case, we will exclude the corresponding bin (which contains this large item) in $\mathsf{OPT}({D'}_q)$ from our analysis. Since we are interested in the second item in $\mathsf{OPT}({D'}_q)$, doing so will exclude at most one such item from our analysis.
		If the same large item had been packed in a bin $B$ other than the last bin in the \ffdq packing of instance ${D'}_{q-1}$, then it would have contained at most two items with size at most $\frac{1}{2}$ and weight at least $W_i$.
		Note that items with size at most $\frac{S}{2}$ and weight at least $W_i$ that are packed in the last bin in the \ffdq packing of ${D'}_{q-1}$, and there is no large item packed in this bin, are analyzed in the proof itself.

	\end{itemize}
	-- In the simple case, if every item of size at most $\frac{S}{2}$ and weight at least $W_i$ from instance $D'$ is packed in a $L$-bin in \ffdqa as a second or third item, then from this and the induction hypothesis, we claim that the number of second items of weight at least $W_i$ in $L$-bins in $\mathsf{OPT}({D'}_q)$ is at most the number of second or third items in $L$-bins in \ffdqa except at most 
	$q-2 < q-1$ 
	item with total weight at most $\frac{q-2}{2} < \frac{q-1}{2}$. \newline
	-- Now, we consider a more complex case. Assume that $d$ is the smallest item (from instance $D'$) not packed in a $L$-bin in \ffdqa, and the size of $d$ is at most $\frac{S}{2}$ and weight at least $W_i$. Consider any bin $B \in P$ that contains a large item from instance $D'$. 
	If no item with size at most $\frac{S}{2}$ and weight at least $W_i$ are packed in this bin, then from our analysis we will exclude the corresponding bin $B$ in $\mathsf{OPT}({D'}_q)$. If this large item had been packed in a bin other than the last bin in the \ffdq packing of ${D'}_{q-1}$, then it would have contained at most two items with size at most $\frac{S}{2}$ and weight at least $W_i$. 
	\begin{itemize}
		\item If $B$ contains a second item of size at least $d$ in $\mathsf{OPT}({D'}_q)$, and if the large item in $B$ is not packed in the last bin in the \ffdq packing of $D_{q-1}$ then the first item in $B$ is of size at most $S-d$, and $B$ also contains a second piece of size at least $d$ in \ffdqa (otherwise, $d$ would have been packed in $B$).
		
		\item[S1] If $B$ contains a second item of size at least $d$ in $\mathsf{OPT}({D'}_q)$, and if the \ffdq packing of same item falls in case K2. Then, we will exclude such a bin from our analysis.

		\item[S2] If in $\mathsf{OPT}({D'}_q)$, bin $B$ contains a second item $c$ of weight at least $W_i$ and size less than $d$, then the following cases can occur:
		\begin{itemize}
			\item Item $c$ has been packed in the last bin in the \ffdq packing of ${D'}_{q-1}$. If there is no large item in this bin then it contradicts our assumption that $d$ is the smallest item not packed in a $L$-bin in \ffdqa.
			
			\item Item $c$ has been packed in the last bin in the 
			\ffdqa packing of ${D'}_{q-1}$, and this last bin also contains a large item other than the one contained in $B$. In such case, we will exclude this bin from our analysis.
		\end{itemize}
		Note that either S1 or S2 will occur. This suggests that we can exclude at most 
		one
		items from our analysis.
		
		\item Otherwise if $\mathsf{OPT}({D'}_q)$ contains a bin $B$ that contains a second item of size less than $d$ and weight at least $W_i$, then by the property of $d$, $c$ must have been packed in a $L$-bin in \ffdqa. Since $x > \frac{2}{11}\cdot S$, $c$ must be either the second or third piece in a $L$-bin in \ffdqa. 
	\end{itemize}
	
	From above analysis, the number of second items of weight at least $W_i$ and size at least $d$ in $\mathsf{OPT}({D'}_q)$ is at most the number of second and third pieces of weight at least $W_i$ and size $d$ in \ffdqa.
	
	\paragraph{Part 2:} Let the second and third items in the $L$-bins in the \ffdqa are arranges in non-increasing order as:
	$p_1, p_2, \ldots p_n$. Similarly, we arrange the second items in $L$-bins in the \optqa packing as:
	$p_1^*, p_2^*, \ldots, p_r^*$.
	In doing so, we exclude at most 
	$q-1$ item from the instance $D_q$ as argued in the above analysis.
	It is easy to verify the claim that for $r \leq n$ and for some $i$, $1 \leq i \leq r$, $W(p_i^*) \leq W(p_i)$ holds. If this is not true, then it implies that the number of second items of weight at least $W(p_i^*)$ is larger than the number of second and third item of weight at least $W(p_i^*)$. Contradicting with what we have proved in the first part of the proof.
	Consequently, it says that the total weight of second items in $L$-bin in $\mathsf{OPT}({D'}_q)$ is at most  the total weight of second and third items in $L$-bin in \ffdqa.
\end{proof}

The idea behind the Baker's proof is to define a weighting function $W$ such that the weight of each bin in the \ffdqa packing is at least one, and at the same time the weight of each bin the $\mathsf{OPT({D'}_q)}$ is at most $\frac{11}{9}$. 

For any weighting function $W$, Baker \cite{Baker_1985} defined $\overline{W} = \sum_{p \in {D'}_q} W(p) - W(2,{L}_P) - W(3,{L}_P)$. This weighting function satisfies the following \cite{Baker_1985}:
\begin{enumerate}
	\item[P1] $\overline{W} > \mathsf{FFDq}({D'}_q) - 3q$ (in \cite{Baker_1985}, it is defined as $\overline{W} > \mathsf{FFDq}({D'}_q) - 3$),
	
	\item[P2] the total weight in each bin $\notin {L}$ in $\mathsf{OPT}({D'}_q)$ is at most $\frac{11}{9}$,
	
	\item[P3] the weight of each first piece in a $L$-bin is $1$,
	
	\item[P4] the weight of third pieces in $L$-bins of $\mathsf{OPT}({D'}_q)$ is no greater than $\frac{2}{9}$,
	
	\item[P5] the total weight of second pieces in $L$-bins in $\mathsf{OPT}({D'}_q)$ is no greater than the total weight of second and third pieces in $L$-bins in \ffdqa.
\end{enumerate}

Collectively, S1-S5 would prove the theorem since:
\begin{align*}
	\mathsf{FFDq}({D'}_q) - 3q &< \overline{W} && \text{From P1 above}\\
	&= \sum_{p \in {D'}_q} W(p) - W(2,{L}_P) - W(3,{L}_P) && \text{From definition of $\overline{W}$} \\
	&\leq \sum_{p \in {D'}_q} W(p) - W(2,{L}_{\mathsf{OPT}({D'}_q)}) && \text{From P5 above} \\
	&\leq \frac{11}{9} \mathsf{OPT}({D'}_q) && \text{From P2,P3,P4 above}
\end{align*}

\begin{proposition} \label{prop: x value range}
	The item $x$ satisfies $\frac{2}{11} \cdot S < x \leq \frac{S}{3}$. 
\end{proposition}
\begin{proof}
	The proof is based on the same line as argued in \cite{Baker_1985}. Here we consider the generalization of it in the context of \qbp. We consider the following cases:
	\begin{itemize}
		\item[(A1)] If $x \leq \frac{2}{11} \cdot S$, then all bins except the $q-1$ previous bins and the last bin have bin size at least $\frac{9}{11} \cdot S$. Reason being either item $x$ did not fit or there is already a copy of item $x$ in the bin packing another $x$ in the same bin will violate \qbp constraint (Note that there $q-1$ such copies of the item $x$). This will yield
		\begin{align*} \label{eqn: eqn1}
			\mathsf{FFDq}({D'}_q) \leq \frac{11}{9}\mathsf{OPT}({D'}_q) + q
		\end{align*}
		
		\item[(A2)] Another simple case could be when $x>\frac{S}{2}$. It is easy to argue that
		\begin{align*} 
			\mathsf{FFDq}({D'}_q) = \mathsf{OPT}({D'}_q)
		\end{align*}
		
		\item[(A3)] If $\frac{S}{2} \geq x > \frac{S}{3}$, for $q=1$ \cite{Baker_1985} has proved that all items that are greater than $S-x \geq S/2$ are packed alone in bins in the packing \ffda and \opta. Also, except for possibly one item all remaining items (with sizes at most $\frac{S}{2}$) are packed in pairs. This fact led them to prove: \ffda = \opta.
		Above relation do not necessarily hold because the first $q-1$ instances in ${D'}_q$ may contain items that are smaller than $x$.\\

		Consider the following weighting scheme:
		\begin{table}[!ht]
			\centering
			\begin{tabular}{c|c|c}
				$i$ & $R_i$ & $W_i$ \\ \hline
				$3$ & $\leq \frac{S}{3}$ & $0$ \\ \hline
				$2$ & $(\frac{S}{3},x]$ & $\frac{1}{3}$ \\ \hline
				$1$ & $(x,\frac{S}{2}]$ & $\frac{S}{2}$ \\ \hline
				$0$ & $(\frac{S}{2},S]$ & $1$
			\end{tabular}
			\caption{Weighting function for Case $\frac{S}{2} \geq x > \frac{S}{3}$}
			\label{tab:weighting scheme}
		\end{table}

		All $L$ bins have a weight of $1$. All bins in $FFDq({D'}_q) \setminus L$ (including the transition bins) except the following bins have weight at least $1$. These bins are as follows:
		\begin{itemize}
			\item there is a transition bin in which the first item is from $R_1$ and other items are from $R_2$.
			\item there is a transition bin in which the first item is from $R_2$ and other items are from $R_3$.
			\item there is a transition bin in which the first item is from $R_3$ and there is at most one fallback piece either from $R_2$ or $R_1$ from the future instance of $D$ in ${D'}_q$.
		\end{itemize}
		Such transition bins exists when packing of the items is going from previous instance of $D$ to the next instance of $D$ (except from the last instance ${D'} \in {D'}_q$). There are at most $3\cdot(q-1)$ such bins. 
		Note that the total weight of the final bin can be less than $1$. Therefore, there are $3\cdot(q-1) + 3 = 3q$ such bins in the \ffdqa packing of ${D'}_q$ than can have weight less than $1$. We will get
		\begin{align*}
			\overline{W}^* \geq FFDq({D'}_q) - 3q
		\end{align*}
		There can be a transition bin from $R_0$ to $R_1$ with weight at least $1$. This proves P1.
		For P2, let's assume that $B \notin L$. The maximum weight of any piece in $B$ is $\frac{S}{2}$. If $B$ contains at most two pieces then $W^*(B) \leq 2 \cdot \frac{S}{2} = S$. If $B$ contains three pieces, $p_1 \leq p_2 \leq p_3$ then $p_3 \leq \frac{S}{3}$ otherwise $p_1 + p_2 + p_3 > 3 \cdot \frac{S}{3} = S$; the weight is at most $\frac{1}{2} + \frac{1}{2} + 0 = 1$. Consequently, P2 holds. 
		It is trivial to see that P3 holds by the virtue of $W^*$. 
		By \Cref{lem: 2 item in optqa at most 2 and 3 item in ffdqa} P5 holds.
		It is easy to see that all the third pieces will be from $R_3$ because if there will be a third piece from either $R_2$ or $R_1$ in $L$-bin then the total sum of all the items will be $> S$. Third pieces from $R_3$ will have a weight of $0$. Therefore, P4 holds. 
		In summary, all of P1-P5 holds, and we conclude that the theorem holds for this case.

		For the remaining paper we can assume that $\frac{2}{11} \cdot S < x \leq \frac{S}{3}$.
		

	\end{itemize}
\end{proof}

For all the cases that we consider in this work, all the weighting functions are a variant of the following function $W^*$.
Let $a = \lfloor S/x \rfloor$.
\begin{align*}
	W^*(z) &= 0 && \text{ if } z < x, \\
		   &= \frac{1}{a} && \text{ if }  z \in \Big[x, \frac{S}{a} \Big] \\
		   &= \frac{1 - W^*(x)}{j} && \text{ if } z \in \Big(\frac{S}{j+1}, \frac{S-x}{j}\Big], j = 2, \ldots, a-1, \\
		   &= \frac{1}{j} && \text{ if } z \in \Big( \frac{S-x}{j}, \frac{S}{j} \Big], j=2, \ldots,a-1, \\
		   &= 1 && \text{ if } z \in (S/2, S]
\end{align*}
Above weighting function is well defined and partitions $[x_s,S]$ \cite{Baker_1985}. 
This weighting function allows every regular bin that packs items with non zero weight(except possibly the last bin and the last bin in the packing of each instance of $D$ in ${D'}_q$) to have a total weight of at least $1$. 
Argument is as follows (due to \cite{Baker_1985}): \\
-- If bin $B$ contains $j$ regular items of weight $1/j$, OR \\
-- If bin $B$ contains $j$ regular items of weight $\frac{1 - W^*(x)}{j}$, then there is a possiblity of fallback item of weight at least $W^*(x)$, resulting in a total weight of at least $1$.

\noindent
Similar to \cite{Baker_1985}, there are four cases to consider:
\begin{enumerate}
	\item[(a)] $\frac{S}{4} < x \leq \frac{S}{3}$,
	\item[(b)] $\frac{2}{11} \cdot S < x \leq \frac{S}{5}$ and there is no regular piece in $\Big(\frac{(S-x)}{4},\frac{S}{4}$ \Big],
	\item[(c)] $\frac{2}{11} \cdot S < x \leq \frac{S}{5}$ and there is a regular piece in $\Big(\frac{(S-x)}{4}, \frac{S}{4} \Big]$, or
	\item[(d)] $\frac{S}{5} < x \leq \frac{S}{4}$
\end{enumerate}
In this paper we analyze cases (a) and (b) only.

\paragraph{Case (a):} $\frac{S}{4} < x \leq \frac{S}{3}$. The weighting function in this case is given by \Cref{tab: weighting function case a}:
\begin{table}[h!]
	\centering
	\begin{tabular}{c|c|c}
		$i$ & $R_i$ & $W_i$  \\ \hline
		$3$ & $(0, x)$ & $0$ \\ \hline
		$2$ & $[x, \frac{(S-x)}{2}]$ & $\frac{1}{3}$ \\ \hline
		$1$ & $(\frac{(S-x)}{2}, \frac{S}{2}]$ & $\frac{1}{2}$ \\ \hline
		$0$ & $(\frac{S}{2}, S]$ & $1$
	\end{tabular}
	\caption{Weighting function: Case (a)}
	\label{tab: weighting function case a}
\end{table}
For each of the first $q-1$ instances of ${D'}_q$ there are at most $2$ transition bins outside of $L$ in the \ffdqa packing of ${D'}_q$. For the first such transition bin the weight is less than $1$. It is possible that the total weight of the second transition bin is less than one. In all, we get at most $2(q-1)$ such transition bins from the first $q-1$ instances of $D$ in ${D'}_q$. The last instance $D'$ in ${D'}_q$ will have one transition bin (that do not contain an item from $R_0$) with weight less than $1$. The total weight of the final bin is also less than one. In all, there are at most $2(q-1) +2 = 2q$ bins with weight less than one. 
This gives:
\begin{align*}
	\overline{W}^* \geq FFDq({D'}_q) - 2q > FFDq({D'}_q) - 3q
\end{align*}
proving P1.
For P2, assume that there is a bin $B$ such that $B \notin L$. Note that the maximum weight of any piece in $B$ is $\frac{S}{2}$. If $B$ contains two items then the total weight is at most $2 \cdot \frac{S}{2} = S$. Assume that $B$ contains three items $p_1, p_2,$ and $p_3$ such that $p_1 \geq p_2 \geq p_3$. We claim that $p_2 \leq \frac{(S-x)}{2}$ because if not, then $p_1 + p_2 + p_3 > 2 \cdot \frac{(S-x)}{2} + x = S$ which is a violation of the bin capacity constraint. With this claim, the weight of the bin is at most $\frac{1}{2} + 2 \cdot \frac{1}{3} = \frac{7}{6}$, proving P2. 
It is trivial to see that P3 holds by the virtue of $W^*$.
By \Cref{lem: 2 item in optqa at most 2 and 3 item in ffdqa},
P5 holds. 
It is easy to check that $L$-bins in $\mathsf{OPT({D'}_q)}$ do not contain third item from $R_1$ or $R_2$, otherwise it will violate the bin capacity constraint. If there is an third item then it will be from $R_3$, and by virtue of $W^*$ P4 is satisfied.
In summary, all of P1-P5 holds in this case. This proves the theorem for this case.

\paragraph{Case (b): } $\frac{2}{11} \cdot S< x < \frac{S}{5}$ and there is no regular piece in $(\frac{(S-x)}{4}, \frac{S}{4}]$.
The weight function in this case is given by \Cref{tab: weighting function for case b}.
\begin{table}[h!]
	\centering
	\begin{tabular}{c|c|c}
		$i$ & $R_i$ & $W_i$ \\ \hline
		$6$ & $< x$ & $0$ \\ \hline
		$5$ & $[x,\frac{S}{4}]$ & $\frac{1}{5}$ \\ \hline
		$4$ & $(\frac{S}{4},\frac{(S-x)}{3}]$ & $\frac{4}{15}$ \\ \hline
		$3$ & $(\frac{(S-x)}{3},\frac{S}{3}]$ & $\frac{1}{3}$ \\ \hline
		$2$ & $(\frac{S}{3},\frac{(S-x)}{2}]$ & $\frac{2}{5}$ \\ \hline
		$1$ & $(\frac{(S-x)}{2},\frac{S}{2}]$ & $\frac{1}{2}$ \\ \hline
		$0$ & $(\frac{S}{2},S]$ & $1$ \\
	\end{tabular}
	\caption{Weighting function for case(b)}
	\label{tab: weighting function for case b}
\end{table}
The weighting function is same as the function given in Baker's paper \cite{Baker_1985} except the addition that assigns weight $0$ to item $< x$. We will prove that the above weighting function results in satisfying all of P1-P5 above. This weighting function assigns weight $W(z) = W^*(z)$ to item $z$ for $z > \frac{S}{4}$, $W(z) = \frac{1}{5}$ for $z \leq \frac{S}{4}$, and $W(z) = 0$ for $z < x$.
We generalize the arguments given in Baker's paper \cite{Baker_1985} for this case. The proof in this case is as follows: \newline
-- A regular bin in \ffdqa that contains pieces from $R_1, R_3$ or $R_5$ contains total weight at least $1$. \newline
-- Now, we will prove that a regular bin that contains items from $R_2$ or from $R_4$ contains a total weight of at least $1$. We will prove this using mathematical induction as follows: \newline
\textbf{Base case:} Let $q=1$. In this case, Baker \cite{Baker_1985} proved that a regular bin that contains items with sizes $>\frac{S}{4}$ has weight at least $1$. \newline
\textbf{Inductive hypothesis:} Let's assume that it is true for $q-1$. \newline
\textbf{Inductive case: } Now, we will prove that it is true for $q$ as well. Consider the packing of items of the instance $D'$ in ${D'}_q$. The only problematic regular bins are the bins that contain items from $R_2$ or $R_4$. 
Consider any bin that contains regular items from $R_2$. There are at most two regular items in that bin with weight at least $\frac{4}{5}$. The remaining space in that bin is at least $x$. Since item $x$ can not be packed into that bin it means the bin contains some item(s) with weight at least $\frac{1}{5}$ resulting in remaining space that can not accommodate item $x$. This gives the total weight of such bin at least $1$.
Same argument can be applied to any regular bin that contains items from $R_4$.
Overall, it proves that each of the regular bins with regular items $> \frac{S}{4}$ contains total weight at least $1$. \newline
-- A regular bin that contains items from $R_5$, by the definition of this case we can say that there are at least five items in the bin each with weight $\frac{1}{5}$, resulting in a total weight of at least $1$. Therefore, we can say that each regular bin \ffdqa contains weight at least $1$.\newline
-- The least density of an item is defined as its weight divided by the maximum item size for that weight, and it is at least $\frac{4}{5}$ \cite{Baker_1985}.
Except at most $q-1$ bins and the last bin, each bin is at least $\frac{4}{5}$ full. Therefore, except at most $q-1$ bins and the last bin each bin contains weight at least $\frac{4}{5} \cdot \frac{4}{5} = \frac{16}{25}$.
Since at most $4q$ transition bins and the last bin can have weight $< 1$,
\begin{align*}
	\overline{W} &\geq \mathsf{FFDq}({D'}_q) - 4q \cdot \frac{9}{25} -1 \\
	&> \mathsf{FFDq}({D'}_q) - 2q -1 \\
	&> \mathsf{FFDq}({D'}_q) - 3q
\end{align*}
Thus, P1 holds. \newline
By the definition of $W$ and \Cref{lem: 2 item in optqa at most 2 and 3 item in ffdqa}.
we say that P3 and P5 hold. 
The third pieces in a $L$-bin are either from $R_5$ or $R_6$ hence P4 holds. \newline 
Now, we will do a case analysis to show that for every bin $B \notin L$, $W(B_{P^*}) \leq \frac{71}{60}$, proving P2. \\
We assume that $B \notin L$. We assume that $B$ contains $i$ pieces $p_1, p_2, \ldots, p_i$ for $1 \leq i \leq 5$ (Note that higher values of $i$ are possible, but in that case there are some items from $R_6$, and from the definition of $W$ their contribution in weight is $0$. Hence, we focus on the cases where there are items having non-zero weights) such that $p_1 \geq p_2 \geq p_3 \geq \ldots \geq p_i$.
\begin{itemize}
	\item \textbf{Case 1,2:} In case, when $B$ has at most two items in $\mathsf{OPT({D'}_q)}$, the total weight is $\leq 2 \cdot \frac{1}{2} = 1$. 
	
	\item \textbf{Case 3:} When there are three items in $\mathsf{OPT({D'}_q)}$. Then,
	\begin{itemize}
		\item If $p_2 \leq \frac{S}{3}$, then the total weight can at most be $\frac{1}{2} + 2 \cdot \frac{1}{3} = \frac{7}{6} < \frac{11}{9}$,
		
		\item If $p_2 > \frac{S}{3}$ and $p_1 \leq \frac{(S-x)}{2}$, since $p_3 \leq \frac{S}{3}$, the total weight is at most $2 \cdot \frac{2}{5} + \frac{1}{3} = \frac{17}{15} < \frac{11}{9}$,
		
		\item If $p_2 > \frac{S}{3}$ and $p_1 > \frac{(S-x)}{2}$, since $p_3 \leq \frac{S}{3}$, if $p_3 > \frac{(S-x)}{3}$, it will violate the bin capacity constraint. Therefore, $p_3 \leq \frac{(S-x)}{3}$. This will give us a total weight of at most $\frac{1}{2} + \frac{2}{5} + \frac{4}{15} = \frac{7}{6} < \frac{11}{9}$
	\end{itemize}
	
	\item \textbf{Case 4:} The Bin $B$ has four items. The following are the analysis of further subcases:
	\begin{itemize}
		\item If $p_2 \leq \frac{S}{4}$, then the total weight of bin $B$ in $\mathsf{OPT({D'}_q)}$ is at most $\frac{1}{2} + 3 \cdot \frac{1}{5} = \frac{11}{10} < \frac{11}{9}$. 
		
		\item Assume that $p_2 > \frac{S}{3}$ and $p_1 > \frac{S}{3}$. Since the minimum value of $p_3$ and $p_4$ can be $x$, the sum of all the item sizes in $B$, $p_1 +p_2 + p_3 + p_4$ will violate the bin capacity constraint. This gives $\frac{S}{4} < p_2 \leq \frac{S}{3}$ ($p_2 \leq \frac{S}{4}$ has been covered above). 
		Now, $p_1$ has to be $\leq \frac{(S-x)}{2}$ otherwise the sum of the item sizes in $B$ of $p_1, p_2, p_3, p_4$ will violate the bin capacity constraint. 
		With this, now the claim is $p_3 \leq \frac{S}{4}$, otherwise the sum of all the items in $B$ will violate the bin capacity constraint. 
		With these constraints on the item sizes, the total weight of all the items in $B$ is at most $\frac{2}{5} + \frac{1}{3} + \frac{2}{5} = \frac{17}{15} < \frac{11}{9}$. 
		Observe that the $p_1 \leq \frac{S}{3}$, this leads to a claim that $p_3 \leq \frac{(S-x)}{3}$, otherwise it will violate the bin capacity constraint. In combination with $p_4 < \frac{S}{4}$, the total weight of bin $B$ is at most $2 \cdot \frac{1}{3} + \frac{4}{15} + \frac{1}{5} = \frac{17}{15} < \frac{11}{9}$.
	\end{itemize}
	\item \textbf{Case 5:} In this case, we can argue that $p_1 \leq \frac{S}{3}$, otherwise with the minimum item size for the pieces $p_1, p_2, p_3, p_4, p_5$, the sum of all item sizes in $B$ will violate the bin capacity constraint. Similarly, it can be argued that $p_2 \leq \frac{S}{4}$, otherwise with constraints on $p_1$, $p_2$ and considering the minimum item sizes for other three items, their sum will violate the bin capacity constraint. 
	This will give us a total weight of at most $\frac{1}{3} + 4 \cdot \frac{1}{5} = \frac{17}{15} < \frac{11}{9}$.
\end{itemize}
$B$ can not contain six items by virtue of $x$. Thus, all of P1-P5 holds in this case. This proves the theorem for this case.
	\section{Conclusion}
\label{sec: conclusion}

In this work, we considered a variant of the bin packing problem in which an item has $q$ number of copies for some positive integer $q>1$ and a bin can contain at most a single copy of an item.
First Fit Decreasing (\FFD) is a classical bin packing problem that can be generalized to solve this problem. 
We have shown that the generalization of \FFD, that is \ffdq, solves \qbp in at most 
$
	\frac{11}{9} \mathsf{OPT(D_q)} + 3q
$
where $q > 1$ for special cases. These special cases are based on the size of the first item in the last bin in the \ffdq packing of the input instance.
Out proofs are based on the proof given in \cite{Baker_1985}.

It will be quite interesting to see whether the same upper bound holds true for the remaining cases. 
Also, it is open whether the given bound is tight. Extending the work by \cite{dosa2007tight, Dósa_Li_Han_Tuza_2013} can be quite useful in this direction.

	\bibliography{refs}
	\bibliographystyle{alpha}	
\end{document}